\documentclass[letterpaper, 10 pt, conference]{ieeeconf}  

\IEEEoverridecommandlockouts                              

\newcommand\todo[1]{\textcolor{black}{#1}}

\usepackage{amssymb}
\usepackage{amsmath}
\usepackage{graphicx}
\graphicspath{{images/}}
\usepackage{comment}
\usepackage[linesnumbered, ruled, vlined]{algorithm2e}
\usepackage{algpseudocode}
\usepackage{float}
\usepackage{multirow}
\usepackage{stackengine}
\usepackage{caption}
\usepackage{subcaption}
\usepackage{array}
\usepackage{lipsum}
\usepackage{wrapfig}
\usepackage{booktabs}
\usepackage{color}
\usepackage{bm}
\usepackage{amsfonts}
\usepackage{pifont}
\usepackage{balance}
\usepackage[english]{babel}

\usepackage{amsthm}
\newtheorem{definition}{Definition}
\newtheorem{remark}{Remark}
\newtheorem{theorem}{Theorem}

\newcommand{\norm}[1]{\left\lVert #1\right\rVert}
\newcommand{\HRule}{\noindent\rule{\linewidth}{0.1mm}\newline}

\DeclareMathOperator*{\argmin}{arg\,min}
\usepackage{xcolor}
\newcommand{\bz}[1]{{\color{black} #1}}
\definecolor{mypink1}{rgb}{0.858, 0.188, 0.478}

\title{\LARGE \bf
Learning Differentiable Safety-Critical Control using Control Barrier Functions for Generalization to Novel Environments
}

\author{Hengbo Ma*,
        Bike Zhang*,
        Masayoshi Tomizuka, 
        and Koushil Sreenath
\thanks{*The authors contributed equally to this work and names are in alphabetical order.}
\thanks{H. Ma, B. Zhang, M. Tomizuka and K. Sreenath are with University of California, Berkeley, CA 94720, USA  (e-mail:hengbo\_ma, bikezhang, tomizuka, koushils@berkeley.edu).}
}

\begin{document}

\maketitle
\thispagestyle{empty}
\pagestyle{empty}

\begin{abstract}
Control barrier functions (CBFs) have become a popular tool to enforce safety of a control system.
CBFs are commonly utilized in a quadratic program formulation (CBF-QP) as safety-critical constraints.
A class $\mathcal{K}$ function in CBFs usually needs to be tuned manually in order to balance the trade-off between performance and safety for each environment.
However, this process is often heuristic and can become intractable for high relative-degree systems.
Moreover, it prevents the CBF-QP from generalizing to different environments in the real world. 
By embedding the optimization procedure of the exponential control barrier function based quadratic program (ECBF-QP) as a differentiable layer within a deep learning architecture,
we propose a differentiable safety-critical control framework that enables generalization to new environments for high relative-degree systems with forward invariance guarantees.
Finally, we validate the proposed control design with 2D double and quadruple integrator systems in various environments. 
\end{abstract}

\section{Introduction}

Safety plays a critical role in autonomous systems that interact with people, such as autonomous driving and robotics.
There are several approaches to developing a safe control strategy, e.g., Hamilton-Jacobi reachability analysis \cite{bansal2017hamilton} and model predictive control \cite{kerrigan2000invariant}. However, such methods may have high computational costs in real-time applications. Control barrier functions (CBFs) \cite{ames2016control} have gained more attention recently since these methods only depend on the current state and do not require heavy computation. 
CBFs are usually encoded as constraints in a quadratic program (CBF-QP) for safety-critical tasks \cite{ames2014control}. 
With a properly chosen class $\mathcal{K}$ function in CBFs, a system can avoid unsafe sets. Meanwhile, it does not reduce the stabilizing performance from a high-level controller \cite{ames2016control}.
However, the performance of the overall controller, which consists of a high-level controller and CBF-QP, can be easily undermined if the environment changes. 
In other words, each safe set in CBFs necessitates a unique class $\mathcal{K}$ function that maximizes the overall control performance for a specific environment.
In the real world, the environment information for safety-critical tasks is usually not fully known a priori, and a system might also face different environments during its deployment. 
Thus, it is hard to tune a class $\mathcal{K}$ function for each environment beforehand to reconcile performance and safety. 
Moreover, the tuning process for choosing a class $\mathcal{K}$ function becomes tedious when there are multiple control barrier function constraints in the CBF-QP \cite{wang2016multi}, or some are with high relative-degree in the ECBF-QP \cite{nguyen2016exponential}.
This challenge impedes the progress towards deploying CBF-based safety-critical controllers in the real world.

To address this challenge, we investigate how to model the relation between environment information and safety-critical control. 
We propose a learning safety-critical control framework using an environment-dependent neural network which satisfies the forward invariance condition.
Thanks to the development of differentiable convex optimization \cite{agrawal2019differentiable}, we can enable the learning procedure in an end-to-end style. After offline training, we can directly deploy the proposed safety-critical control framework in different environments without any adaption.

\subsection{Related Work}

\subsubsection{Safe environment generalization}
Cluttered environments have been considered in the safe control literature.
A provably approximately correct-bayes framework is proposed to synthesize controllers that provably generalize to novel environments in \cite{majumdar2021pac}.
A control Lyapunov function and control barrier function based quadratic program (CLF-CBF-QP) is utilized with a high-level path plan to navigate through obstacle-scattered environments in \cite{barbosa2020provably}.
Moreover, for hostile environments with adversarial agents, a probabilistic tree logic method is proposed in \cite{cizelj2011probabilistically} to assure safety.
Safe generalization problem with control barrier functions is considered with a weighted mixture of existing controllers in \cite{song2021generalization}. Yet the generalization ability is largely limited by the number of existing controllers.
With the help of reinforcement learning, safe environment adaptation is studied through a risk-averse approach in \cite{zhang2020cautious}. However, this approach can only provide relative safety instead of safety guarantee.
For robotic applications, bipedal robot walking on stepping stones is addressed in \cite{nguyen2015optimal} using a robust control barrier function method, where the distances between adjacent stones are different at each step. 
Predictive control with CBFs tackles the safe car overtaking problems in \cite{zeng2021mpccbf}, where different leading cars serve as novel environments.
Our approach adopts a different way using class $\mathcal{K}$ function to generalize a controller to enforce safety under different environments.

\subsubsection{Safe learning control}
A safe reinforcement learning (RL) framework under constrained Markov decision process is proposed in \cite{chow2018lyapunov} using a Lyapunov based method.
A learning-based control barrier function from expert demonstration is proposed in \cite{robey2020learning} to ensure safety.
In \cite{scukins2021using}, a CBF is created using RL for risk mitigation in adversarial environments.
In \cite{choi2020reinforcement} and \cite{taylor2020learning}, they address the model uncertainty problem by learning CBF constraints.
In \cite{dawson2021safe}, the authors design a learning robust control Lyapunov barrier function that can generalize despite model uncertainty.
A model-free safe reinforcement learning is studied by synthesizing a barrier certificate and querying a black-box dynamic function in \cite{zhao2021model}.
A game theoretic approach is adopted in \cite{tian2021safety} to reduce conservatism while maintaining robustness during human robot interaction.
Differentiable optimization layers have emerged as a new approach for safe learning control recently.
In \cite{parwana2021recursive}, a differentiable layer is applied to control barrier function based quadratic program in order to enhance the recursive feasibility, where the parameters are adapted online.
In \cite{emam2021safe}, safety is framed as a differentiable robust CBF layer in model-based RL.
We also utilize the differentiable optimization layer as a tool. 
However, we focus on generalizing the safety-critical control to novel environments.

\subsection{Contributions}
The contribution of this paper is as follows:
\begin{itemize}
    \item We present an approach to generalizing safety-critical control to novel environments by integrating control barrier functions and differentiable optimization.
    \item We introduce a neural network based ECBF-QP and formulate the safety-critical control as a differentiable optimization layer.
    \item We show that the proposed neural network module based on the exponential control barrier function assures the forward invariance of a safe set.
    \item We numerically validate the proposed learning control design using systems with different relative-degrees and novel environments with randomly generated obstacles.
\end{itemize}

\subsection{Organization}
This paper is organized as follows: in Sec.~\ref{sec:background}, we introduce the background of control barrier functions and differentiable optimization. 
The problem formulation is illustrated in Sec.~\ref{sec:statement}, where we motivate the formulation with a simple case study.
Then, in Sec.~\ref{sec:method}, we present the methodology of learning differentiable safety-critical control using control barrier functions. 
In Sec.~\ref{sec:experiment}, we test the proposed control logic on 2D double and quadruple integrator systems with different environment settings. 
Secs.~\ref{sec:discussion} and \ref{sec:conclusion} provides discussion and concluding remarks.
\section{Background}
\label{sec:background}

Throughout this paper, we will consider a nonlinear control affine system:
\begin{equation}
    \dot{x} = f(x) + g(x)u,
    \label{eq:system}
\end{equation}
where $x \in \mathcal{X} \subset \mathbb{R}^n$ represents the state of the system, $u \in \mathbb{R}^m$ is the control input, and $f : \mathcal{X} \rightarrow \mathbb{R}^{n}$ and $g: \mathcal{X} \rightarrow \mathbb{R}^{m}$ are locally Lipschitz continuous.

\subsection{Control Barrier Functions}
\begin{definition}
\label{def:class_K}
\cite{khalil2002nonlinear} A Lipschitz continuous function $\alpha: [0,a) \rightarrow [0, \infty), a > 0$ is said to belong to class $\mathcal{K}$ if it is strictly increasing and $\alpha(0) = 0$.
Moreover, $\alpha$ is said to belong to class $\mathcal{K}_\infty$ if it belongs to class $\mathcal{K}$, $a = \infty$, and $\lim_{r\to\infty} \alpha(r) = \infty$.
\end{definition}
\begin{definition}
\label{def:cbfs}
\cite[Def.~2]{ames2019control} Consider a continuously differentiable function $h: \mathcal{X} \subset \mathbb{R}^n \to \mathbb{R}$ and a set $\mathcal{C}$ defined as the superlevel set of $h$, $\mathcal{C} = \{ x \in \mathcal{X}: h(x) \geq 0 \}$, then $h$ is a control barrier function (CBF) if there exists an extended class $\mathcal{K}_\infty$ function $\alpha$ such that for the control system \eqref{eq:system}:
\begin{equation}
    \sup_{u \in \mathbb{R}^m} [L_f h(x) + L_g h(x) u] \geq - \alpha (h(x)).
    \label{eq:cbf_def}
\end{equation}
\label{def:cbf}
\end{definition}
\vspace{-5mm}
\noindent
If $h$ is a control barrier function on $\mathcal{X}$ and $\frac{\partial h}{\partial x} \neq 0$ for all $x \in \partial \mathcal{C}$, any Lipschitz continuous controller satisfying \eqref{eq:cbf_def} renders the set $\mathcal{C}$ forward invariant \cite[Thm.2]{ames2019control}. By incorporating \eqref{eq:cbf_def} as a constraint, a quadratic program based safety-critical controller is proposed in \cite{ames2014control}:
\HRule
\noindent \textbf{CBF-QP}:
\begin{subequations}
\label{eq:cbf-qp-all}
\begin{align}
u^{*}(x) & = & & \underset{u\in \mathbb{R}^{m}}{\argmin}  \quad \norm{u - u_\text{perf}}^2 \label{eq:cbf-qp}\\
& \text{s.t.} & & L_f h(x) + L_g h(x)u \geq - \alpha (h(x)), \label{eq:cbf-constraint}
\end{align}
\end{subequations}
\hrule
\vspace{2mm}
\noindent
where $u_\text{perf}$ is the reference control input that can be from a high-level performance controller, which is expected to achieve the control objective.
For instance, model predictive control is a popular choice as a high-level performance controller \cite{rosolia2020multi}.
In the context of safety-critical control, a control Lyapunov function is often used in a quadratic program formulation (CLF-QP) to realize stability.

\begin{remark}
In Definition \ref{def:cbfs}, an extended class $\mathcal{K}_\infty$ function is required for CBF.
Here, we restrict ourselves to a subclass: class $\mathcal{K}$ function, which can facilitate our learning algorithm.
Typically, $\alpha (x)$ is simplified as $\alpha x$, with $\alpha$ being a positive constant, which we term as a linear class $\mathcal{K}$ function.
Previous work \cite{parwana2021recursive, xiao2019feasibility, zeng2021safety} have investigated how to adjust the class $\mathcal{K}$ function in order to improve the feasibility.
In this work, we focus on learning a neural network based class $\mathcal{K}$ function to safely generalize to different environments.
\end{remark}

The CBF constraint in \eqref{eq:cbf-constraint} has been so far assumed to be relative-degree one, which typically does not held for most safety-critical constraints in robotic systems \cite{hsu2015control}.
A special type of CBFs called exponential control barrier functions (ECBFs) has been introduced to enforce arbitrarily high relative-degree CBF constraints in \cite{nguyen2016exponential}.

\begin{definition}
\cite[Def.~1]{nguyen2016exponential} Consider a r-times continuously differentiable function $h: \mathcal{X} \subset \mathbb{R}^n \to \mathbb{R}$ and a set $\mathcal{C}$ defined as the superlevel set of $h$, $\mathcal{C} = \{ x \in \mathcal{X}: h(x) \geq 0 \}$, then $h$ is an exponential control barrier function (ECBF) if there exists a row vector $K_\alpha \in \mathbb{R}^{r}$  such that for the system \eqref{eq:system}:
\begin{equation}
    \sup_{u \in \mathbb{R}^m} [L_f^r h(x) + L_g L_f^{r-1} h(x) u] \geq - K_\alpha \eta_b(x),
    \label{eq:ecbf_def}
\end{equation}
\label{def:ecbf}
for $\forall x \in \{x \in \mathbb{R}^n | h(x) \geq 0\}$, with
\begin{equation}
	\begin{aligned}
		\eta_b(x) &=
		\begin{bmatrix}
			h(x)\\
			\dot{h}(x)\\
			\ddot{h}(x)\\
			\vdots \\
			h^{(r-1)}(x)
		\end{bmatrix}
		\hspace{-2mm}
		&=
		\begin{bmatrix}
			h(x)\\
			L_{f}h(x)\\
			L_{f}^2h(x)\\
			\vdots \\
			L_{f}^{r-1}h(x)
		\end{bmatrix}.
	\end{aligned}
\end{equation}
\end{definition}
\noindent
We define $\mu = L_f^rh(x) + L_g L_f^{r-1} h(x) u$, then the above dynamics of $h(x)$ can be written as the linear system
\begin{equation}
    \begin{aligned}
    \dot{\eta_b}(x) & = F \eta_b(x) + G  \mu, \\
    h(x) & = C \eta_b(x),
    \end{aligned}
\end{equation}
where
\begin{equation}
	\begin{aligned}
		F &=
		\begin{bmatrix}
			0 & 1 & 0 & \dots & 0 \\
			0 & 0 & 1 & \dots & 0 \\
			\vdots & \vdots & \vdots & \ddots & \vdots \\
			0 & 0 & 0 & \dots & 1 \\
			0 & 0 & 0 & \dots & 0
		\end{bmatrix},
		\hspace{3mm}
		G =
		\begin{bmatrix}
			0 \\
			0 \\
			\vdots \\
			0 \\
			1
		\end{bmatrix}, \\
		C &=
		\begin{bmatrix}
			1 & 0 & \dots & 0
		\end{bmatrix}.
	\end{aligned}
\end{equation}
If $\mu \ge -K_\alpha \eta_b(x)$, with $(F-GK_\alpha)$ being Hurwitz and total negative, then we can guarantee that $h(x_0) \geq 0$ $\implies$ $h(x(t)) \geq 0, \forall t \geq 0$ where $x_0$ is the initial condition.

Let $-p_i$ be the negative real eigenvalues of $(F-GK_\alpha)$. We can then define a family of functions $v_i: \mathcal{X} \subset \mathbb{R}^n \to \mathbb{R}$ with corresponding superlevel sets $\mathcal{C}_i$,
\begin{equation}
	\begin{aligned}
		v_0(x) & = h(x),                         && \mathcal{C}_0 = \{x: v_0(x) \geq 0 \}, \\
        v_1(x) & = \dot{v}_0 + p_1 v_0 (x),      && \mathcal{C}_1 = \{x: v_1(x) \geq 0 \}, \\
               & \vdots && \vdots \\
        v_r(x) & = \dot{v}_{r-1} + p_r v_{r-1} (x),  &&\mathcal{C}_r = \{x: v_r(x) \geq 0 \},
	\end{aligned}
\label{equ: ecbf_eigen}
\end{equation}
where $\mathcal{C}_0$ plays the role of the safe set $\mathcal{C}$ as defined in Definition \ref{def:cbf} for a relative-degree one CBF.  
Then, we have:
\bz{
\begin{theorem}
\cite[Thm.2]{nguyen2016exponential} A valid exponential CBF should satisfy two conditions: suppose $K_\alpha$ is chosen such that $p_i > 0$ and the eigenvalues $-p_i$ satisfy $p_i \geq - \frac{\dot{v}_{i-1}(x_0)}{v_{i-1}(x_0)}$,
then \eqref{eq:ecbf-constraint} guarantees $h(x)$ is an exponential control barrier function.
\label{thm:ecbf}
\end{theorem}
}

Given an ECBF, we can extend the CBF-QP in \eqref{eq:cbf-qp-all} to enforce high relative-degree safety-critical constraints:
\HRule
\noindent \textbf{ECBF-QP}:
\begin{subequations}
\label{eq:ecbf-qp-all}
\begin{align}
u^{*}(x) & = & & \underset{u\in \mathbb{R}^{m}}{\argmin}  \quad \norm{u - u_\text{perf}}^2 \label{eq:ecbf-qp}\\
& \text{s.t.} & & L_f^r h(x) + L_g L_f^{r-1} h(x) u \geq - K_\alpha \eta_b(x), \label{eq:ecbf-constraint}
\end{align}
\end{subequations}
\hrule
\vspace{2mm}
\noindent
where $u_\text{perf}$ is the reference control input. 
Note that the control barrier function constraint \eqref{eq:ecbf-constraint} can be extended to multiple constraints in order to account for different safety criteria.
Furthermore, a general formulation of high order control barrier functions can be seen in \cite{xiao2021high}.
\begin{remark} \label{rmk:Lf_operator}
In ECBF-QP \eqref{eq:ecbf-constraint}, due to the relation between the coeeficients in $K_\alpha$ and the eigenvalues $-p_i$, $K_\alpha \eta_b$ can be reformulated as $\Pi_{i=1}^{r} [L_f + \text{p}_i] \circ h(x) - L_f^{r}h(x)$. This will be used to develop our differentiable safety-critical control formulation. Note that $L_f$ here is the Lie derivative operator s.t. $L_f \circ h(x) = L_f h(x)=\frac{\partial}{\partial x}h(x)f(x)$.
\label{remark: Ka}
\end{remark}

\subsection{Differentiable Optimization}
A differentiable optimization problem is a class of optimization problems whose solutions can be backpropagated through.
This functionality enables an optimization problem to serve as layers within deep learning architectures, which can encode constraints and complex dependencies through optimization that traditional convolutional and fully-connected layers usually cannot capture \cite{amos2017optnet}.
Some successful differentiable layer examples include differentiable model predictive control \cite{amos2018differentiable, drgona2020learning}, \bz{Pontryagin's maximum principle \cite{jin2021safe}, robust control \cite{donti2020enforcing}}, and meta learning \cite{lee2019meta}, etc.
In this work, we utilize the differentiable optimization layer presented in \cite{agrawal2019differentiable} for the ECBF-QP.
\begin{remark}
While CBFs are continuously differentiable functions \cite{ames2019control}, here a differentiable safety-critical control using CBF does not mean the CBF itself is differentiable but rather that the backpropagation can go through the CBF-QP differentiable optimization layer.
\end{remark}
\section{Problem Statement}
\label{sec:statement}

Having established the background of CBFs and differentiable optimizations, we now present our problem formulation for generalizing to novel environments.

\begin{figure}
\centering
    \begin{subfigure}[r]{0.49\columnwidth}
        \includegraphics[width=\textwidth]{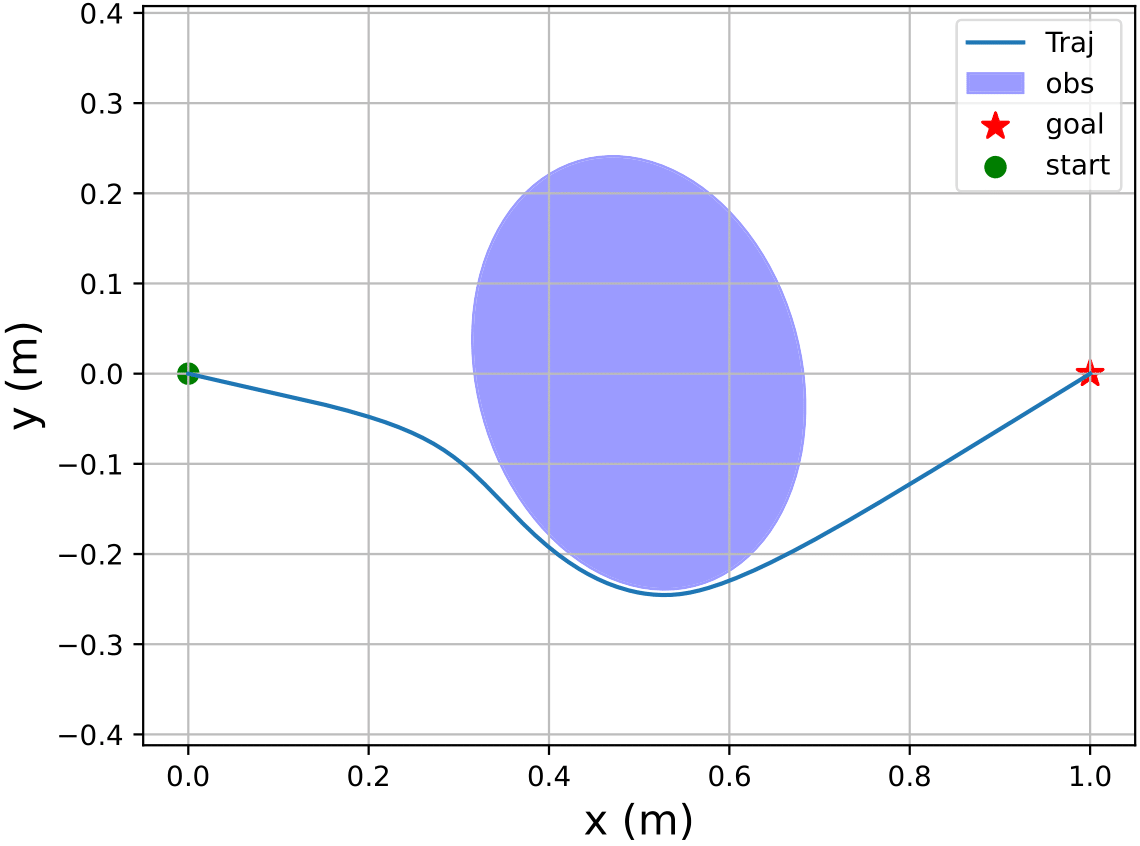}
    	\caption{Environment 1}
    \end{subfigure}
    \begin{subfigure}[r]{0.49\columnwidth}
        \includegraphics[width=\textwidth]{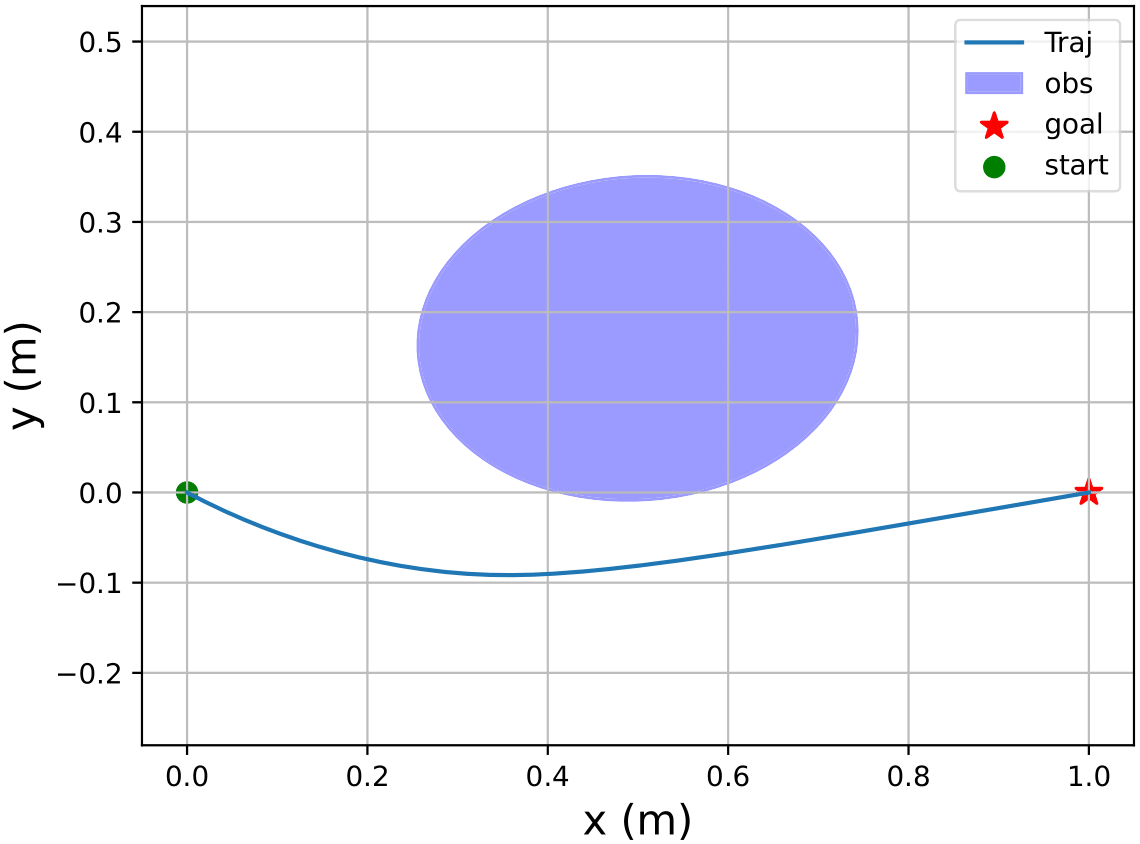}
    \caption{Environment 2}
    \end{subfigure}
    \caption{Motivating example for safety-critical control for generalization to novel environments using a 2D double integrator. A hand-tuned $K_\alpha$ for Environment 1 in (a) is used in the novel Environment 2 in (b). As can be seen, this results in a trajectory with a larger deviation from the obstacle in Environment 2. Thus, a well-tuned $K_\alpha$ for one environment does not necessarily generalize to a different environment.}
    \label{fig:motivating_example}
\end{figure}

\subsection{Motivating Example}

Using a 2D double integrator as an illustrative example, we design a linear quadratic regulator (LQR) to drive the system to a goal location while avoiding different obstacles using the ECBF-QP in \eqref{eq:ecbf-qp-all}.
The LQR controller serves as the reference performance controller $u_\text{perf}$.
The simulation results are demonstrated in Fig.~\ref{fig:motivating_example}.
Note that the $K_\alpha$ for ECBF-QP is tuned manually in Fig.~\ref{fig:motivating_example}(a), which leads to a short and smooth trajectory, i.e., a smooth trajectory that goes around the obstacle with minimal detour from a straight-line trajectory from start to goal. 
However, the ECBF-QP with the same $K_\alpha$ results in a large detour in Fig.~\ref{fig:motivating_example} (b) in a different environment.
While larger detours are conservative, they potentially require more control effort, and result in energy inefficient motions.
The motion in Environment 2 could be shorter by getting closer to the obstacle.
This example demonstrates that the $K_{\alpha}$ plays an important role in generating desirable trajectories in different environments.

\begin{remark}
In order to get a trajectory with desired properties, e.g., smoothness and minimum distance, it is necessary to choose a proper $K_{\alpha}$ using ECBF-QP.
Moreover, certain fixed $K_{\alpha}$ that works in a particular environment could actually fail in a different environment, resulting in violation of the safety constraint $h(x)\ge0$.
\end{remark}

\begin{figure*}[t]
    \centering
    \includegraphics[width=1.0\linewidth]{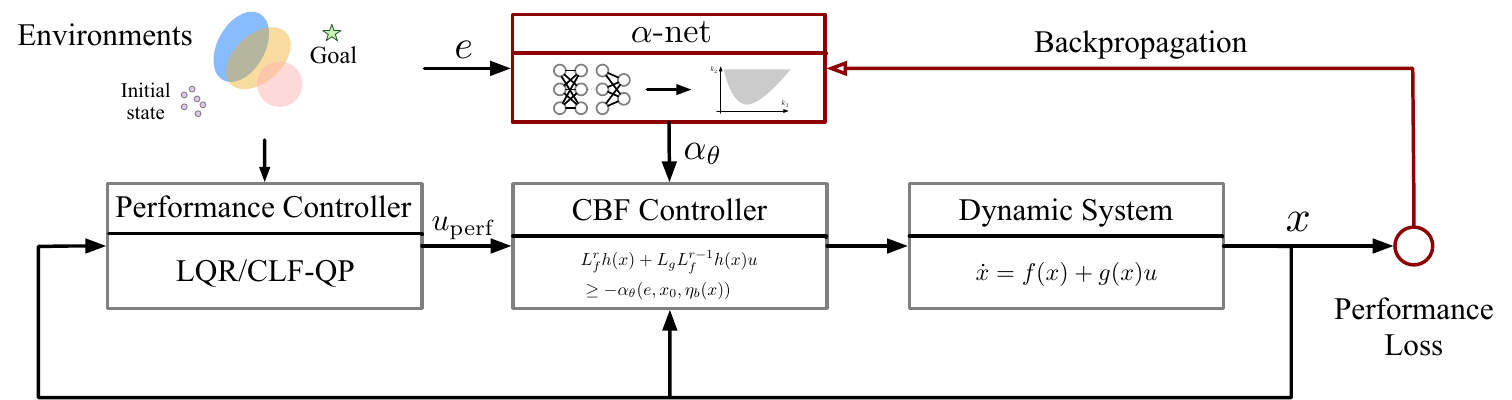}
    \caption{The overall framework of the proposed approach, which includes two main components: a performance controller and a differentiable CBF-QP. The novel environment information, $e$, is an input to the performance controller and $\alpha$ net. The performance loss computed along a trajectory will be backpropagated through the $\alpha$ net, then the $\alpha$ net outputs the parameters to construct the class $\mathcal{K}$ function.}
    \label{fig:architecture}
\end{figure*}

\subsection{Problem Formulation}
Building upon the motivating example, we are motivated to optimize the class $\mathcal{K}$ function in CBF-QP or $K_\alpha$ in ECBF-QP with respect to different environments, which can result in a safe trajectory satisfying a user-defined metric.

To this end, we represent the $K_{\alpha}\eta_{b}(x)$ in \eqref{eq:ecbf_def} with a neural network parameterized with $\theta$. $\Pi_{\theta}(u_{\text{perf}}, e, x_0)$ represents the solution of ECBF-QP mentioned in \eqref{eq:ecbf-constraint}. Such an ECBF-QP can be embedded as a layer in a deep learning pipeline by using differentiable convex optimization technique. Then we minimize a performance cost $\mathcal{L}$ in an episodic setting. The formulation is given as follows:
\begin{equation}
\begin{aligned}
& \underset{\theta}{\argmin} \quad \mathbb{E}_{x_0 \sim P_{0}, e \sim P_{e}}[\mathcal{L}(\tau, e, x_0)]\\
\text{s.t.} \quad & u = \Pi_{\theta}(u_{\text{perf}}(x), e, x_0), \\
& \dot{x} = f(x) + g(x)u,\\
\end{aligned}
\end{equation}
where $e$ is an environment sampled from a distribution of environments $P_{e}$, e.g., $e$ consists of center and radius of the obstacle.
$x$ is the state where we evaluate cost at each time step, and $x_0$ is the initial state which is sampled from a distribution $P_{0}$.
Note that $\mathcal{L}$ is the cost along a trajectory instead of the cost at each time step, and $\tau$ represents the trajectory with time horizon $T$.
$u_{\text{perf}}$ is the performance control input provided by a high-level performance controller.
Once the training procedure is done offline, we can deploy the neural network based controller $\Pi_{\theta}$ to novel environments.

\section{Methodology}
\label{sec:method}

Having seen the problem formulation, we will next introduce how to enable the generalization to novel environments via learning differentiable ECBF-QP.

The overall control architecture is shown in Fig.~\ref{fig:architecture}, which basically includes two parts: a performance controller and a differentiable ECBF-QP safety filter.
The performance controller is mainly responsible for achieving the control objective, and the differentiable ECBF-QP serves as a safety filter, which will be explained in detail in this section.

\subsection{Differentiable Safety-Critical Control using CBFs}

We formulate our differentiable safety-critical control based on exponential control barrier functions in \eqref{eq:ecbf-qp-all}.
Differentiable CBFs have been used in \cite{parwana2021recursive} and \cite{emam2021safe}. 
However, they used systems with relative-degree one or solved a relative-degree two system using a cascaded approach. 
Here, we extend it to a general formulation as follows:
\HRule
\noindent \textbf{Differentiable  ECBF-QP}:
\begin{equation}
    \label{equ:diff-cbf}
	\begin{aligned}
		&\Pi_{\theta}(u_{\text{perf}}, e, x_0) = 
		\underset{u\in \mathbb{R}^{m}, \delta \in \mathbb{R}}{\argmin} ~ \|u-u_{\text{perf}}\|^2 + \zeta \delta^2 \\
		\text{s.t.} \quad &L_{f}^{r}h(x) + L_{g}L_{f}^{r-1}h(x)u \ge - \alpha_{\theta}(e, x_0, \eta_b(x)) - \delta^2,
	\end{aligned}
	\vspace{1.5mm}
\end{equation}
\hrule
\vspace{2mm}
\noindent
where, $\Pi_{\theta}(u_{\text{perf}}, e, x_0)$ is the safe policy filtered by the ECBF-QP and conditioned on the high-level performance control input $u_{\text{perf}}$ and the environment information $e$. 
$\alpha_{\theta}(e, x_0, \eta_b(x))$ is denoted by $\alpha$-net, where $\theta$ represents parameters of the network.

As we will see next, the $\alpha$ function is a linear function (of $\eta_b$) that encodes an ECBF constraint within the differentiable ECBF-QP so as to deal with high relative-degree safety constraints, which are common in many robotic applications. We include a slack variable $\delta$ which guarantees that such an optimization is feasible during the training procedure, and $\zeta$ is a hyperparameter. Note that we do not use $\delta$ as in \eqref{equ:diff-cbf} during the test time.

\subsection{The Structure of $\alpha$-net}
Exponential control barrier function provides a formal structure to guarantee safety with a \bz{vector parameter $K_\alpha$}. 
In general, it is not easy and probably time-consuming to find the best $K_{\alpha}$ directly in order to generalize to novel environments.
We thus encode the structure of exponential CBF into our neural network. As noted in Remark \ref{remark: Ka}, our formulation is shown as follows
\begin{equation}
\begin{aligned}
&\alpha_{\theta}(e, x_0, \eta_b(x)) = \prod_{i=1}^{r}[L_f + p_{i}(e, x_0;\theta)]\circ h(x) - L_f^{r}h(x), \\
\end{aligned}
\end{equation}
\noindent
where $L_f$ is the lie derivative operator as mentioned in Remark \ref{rmk:Lf_operator}. 
Notice that for the ECBF, the right side of \eqref{eq:ecbf-constraint} only includes the lie derivative with respect to $f$. The function $\textbf{p}(e, x_0; \theta) \in \mathbb{R}^{r}$ outputs $[p_1, p_2, \dots, p_r]$, and $p_{i}(e, x_0; \theta)$ represents $p_i$ as defined in \eqref{equ: ecbf_eigen}. We use the following neural network structure:
\begin{equation}
\begin{aligned}
& \textbf{p}(e, x_0; \theta) =\text{ReLU}(\prod_{k=0}^{m}\sigma(W_kl_k)-b(x_0)) + b(x_0), \\
& b_{i}(x_0) = \text{ReLU} (-\frac{\dot{v}^{i-1}(x_0)}{v^{i-1}(x_0)}-\epsilon) + \epsilon, i=1, \dots, r,\\
& l_0 = [e, x_0], \\
\end{aligned}
\label{eq:inavriance_condition}
\end{equation}

\noindent
where $b_{i}(x_0)$ is the $i$-th element of $b(x_0)\in \mathbb{R}^{r}$, and it represents the bounds of $p_{i=1 \dots r}$ in Thm.\ref{thm:ecbf}. The parameter $\theta$ represents the weights $\{W_0, W_1, \dots, W_m\}$. $l_k$ represents the outputs of the $k$-th layer of the neural network. We concatenate the environment information $e$ and initial state $x_0$ together as the input $l_0$ and choose the ReLU function as the activation, with $\sigma(\cdot)$ being any activation function.
Then, we randomly initialize the neural network parameters with positive numbers.
\todo{
\begin{theorem}
\label{thm:nn_ecbf}
Given the function $\textbf{p}(e, x_0; \theta)$ defined in \eqref{eq:inavriance_condition}, $p_i$ satisfies the conditions in Thm.\ref{thm:ecbf}. Thus, $h(x)$ is an exponential CBF and $\mathcal{C}_{0}$ in \eqref{equ: ecbf_eigen} is forward invariant.
\end{theorem}
\begin{proof}
Since $b_{i}(x) = \max\{-\frac{\dot{v}^{i-1}(x_0)}{v^{i-1}(x_0)}, \epsilon\}$ and $p_{i}(e, x_0, \theta) \ge  b_{i}(x_0)$, we have $p_{i}(e, x_0, \theta) \ge \max \{-\frac{\dot{v}^{i-1}(x_0)}{v^{i-1}(x_0)}, \epsilon\}$. It follows that $p_{i}(e, x_0, \theta)$ satisfies i) $p_{i}(e, x_0, \theta) \ge -\frac{\dot{v}^{i-1}(x_0)}{v^{i-1}(x_0)}$ and ii) $p_{i}(e, x_0, \theta) \ge \epsilon > 0$, which are the conditions in Thm.\ref{thm:ecbf}. From \cite[Thm.1]{nguyen2016exponential}, the set $\mathcal{C}_{0}$ is forward invariant given $h(x)$ is a valid exponential CBF.
\end{proof}
}

\subsection{Loss Function}
In general, the loss function $\mathcal{L}(\tau, e, x_0)$ can be designed with any performance evaluation metric $\mathcal{L}_{\text{perf}}$. In our work, we propose a loss function which includes two components:
\begin{equation}
\begin{aligned}
&\mathcal{L}(\tau, e, x_0) = \mathcal{L}_{\text{perf}}(\tau, e, x_0) + \lambda_{\delta} \sum_{t=1}^{T}\delta^2_{t}, \\
\end{aligned}
\label{eq:loss}
\end{equation}
where $T$ is the number of simulation time steps with a fixed simulation interval $\Delta t$, $\tau$ is the simulated trajectory represented by $[x_0, x_1, \dots, x_T]$.
The coefficient $\lambda_{\delta}$ is for slack variable penalty.
Notice that we use a slack variable $\delta$ in \eqref{equ:diff-cbf} to make sure that the optimization program will not be interrupted by the infeasibility issue of solving the ECBF-QP. Here, $\delta_{t}$ represents the value of the slack variable $\delta$ at each time step. However, the gradient descent of the neural network $\alpha_\theta$ may lead to a solution such that $\delta_t$ is large. Hence, we include the penalty of $\delta_t$ in the loss function. The ideal situation is that $\delta_t$ is zero. 
\vspace{-0.1cm}
\subsection{Algorithm}
The training algorithm is shown in Algorithm~1.
The input is a distribution of environments $P_e$, and the output is the network weights $\theta$ of the $\alpha$-net.
For each iteration, we sample $n$ environments and rollout trajectories $\tau$, then the weights of the neural network are updated after each iteration.

\SetKwInput{KwInput}{Input}                
\SetKwInput{KwOutput}{Output} 
\begin{algorithm}]
	\DontPrintSemicolon
	\KwInput{Environment distribution $P_e$, initial state distribution $P_0$, simulation time interval $\Delta t$, simulation time horizon $T$.}
	\KwOutput{The network weights $\theta$.}
	\While{t $\leq$ number of iteration}{
		\For{i=1:n}    
		{ 
			Sample $e_i \sim P_e$, $x_{0,i} \sim P_0$\\
			Collect the trajectory $\tau_i$ by using the designed controller.
		}
		Update $\theta_{t} \rightarrow \theta_{t-1} - \lambda\frac{1}{n}\nabla_{\theta}\sum_{e_i}\mathcal{L}_{\theta}(\tau_i, e_i, x_{0,i})$
	}
	\caption{Training algorithm}
	\label{algo}
\end{algorithm}
\noindent
When the task is obstacle avoidance, we can iteratively use the learned policy for $m$ obstacles during the test time. The differentiable ECBF-QP in \eqref{equ:diff-cbf} becomes
\begin{equation}
    \label{equ:diff-cbf-1}
	\begin{aligned}
		& \Pi_{\theta}(u_{\text{perf}}, e, x_0) = \underset{u\in \mathbb{R}^{m}}{\argmin} ~ \|u-u_{\text{perf}}\|^2\\
		\text{s.t.} & \quad L_{f}^{r}h_j(x) + L_{g}L_{f}^{r-1}h_j(x)u \ge -\alpha_{\theta}(e_j, x_0, \eta_{b,j}(x)), \\ 
		& \quad j=1, \dots, m,\\
	\end{aligned}
\end{equation}
where $h_j$ represents the $j$-th exponential CBF. The environment $e$ can have multiple obstacles and each of them $e_j$ can be captured by an ECBF constraint. 
\section{Results}
\label{sec:experiment}

After developing our methodology for learning differentiable safety-critical control using CBFs, we now present the simulation results of our proposed framework using 2D double and quadruple integrator systems.

\begin{figure}
\centering
    \begin{subfigure}[r]{0.49\columnwidth}
        \includegraphics[width=\textwidth]{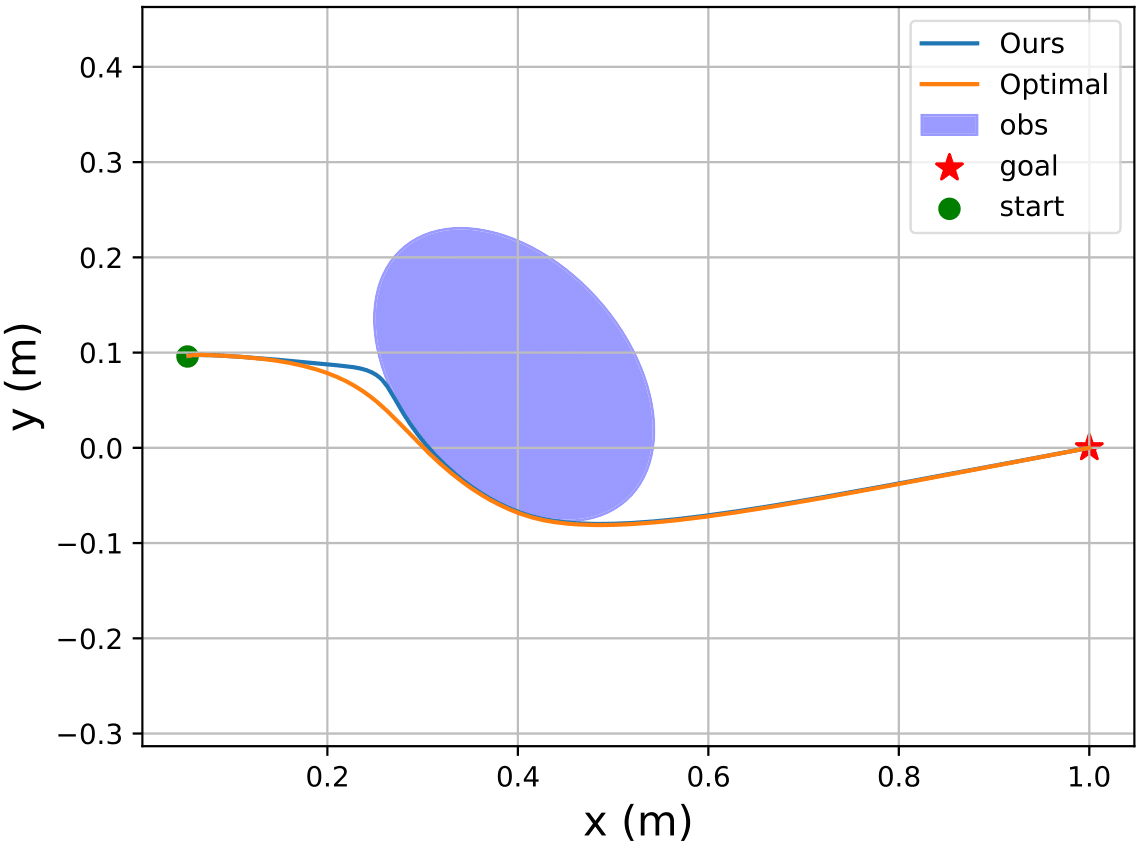}
    	\caption{Environment 1: Trajectory}
    \end{subfigure}
    \begin{subfigure}[r]{0.49\columnwidth}
        \includegraphics[width=\textwidth]{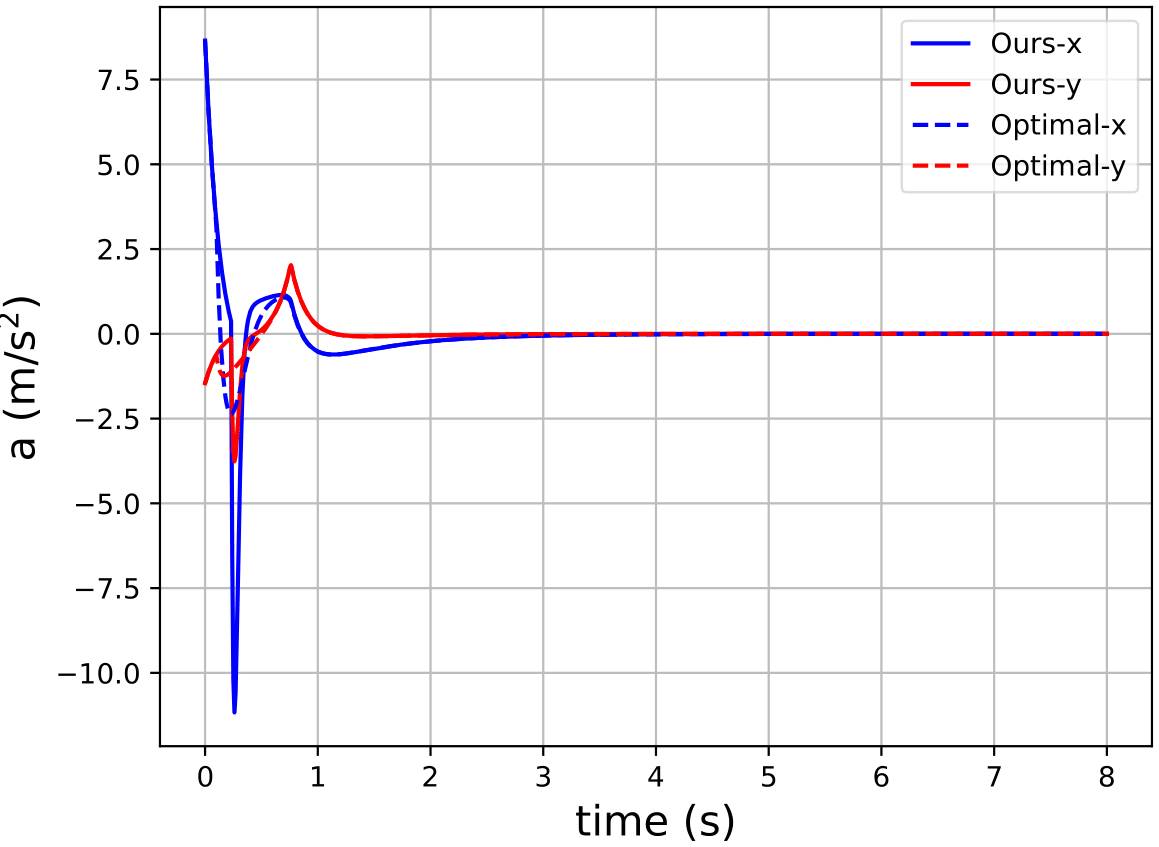}
    \caption{Environment 1: Control input}
    \end{subfigure}
    \begin{subfigure}[r]{0.49\columnwidth}
        \includegraphics[width=\textwidth]{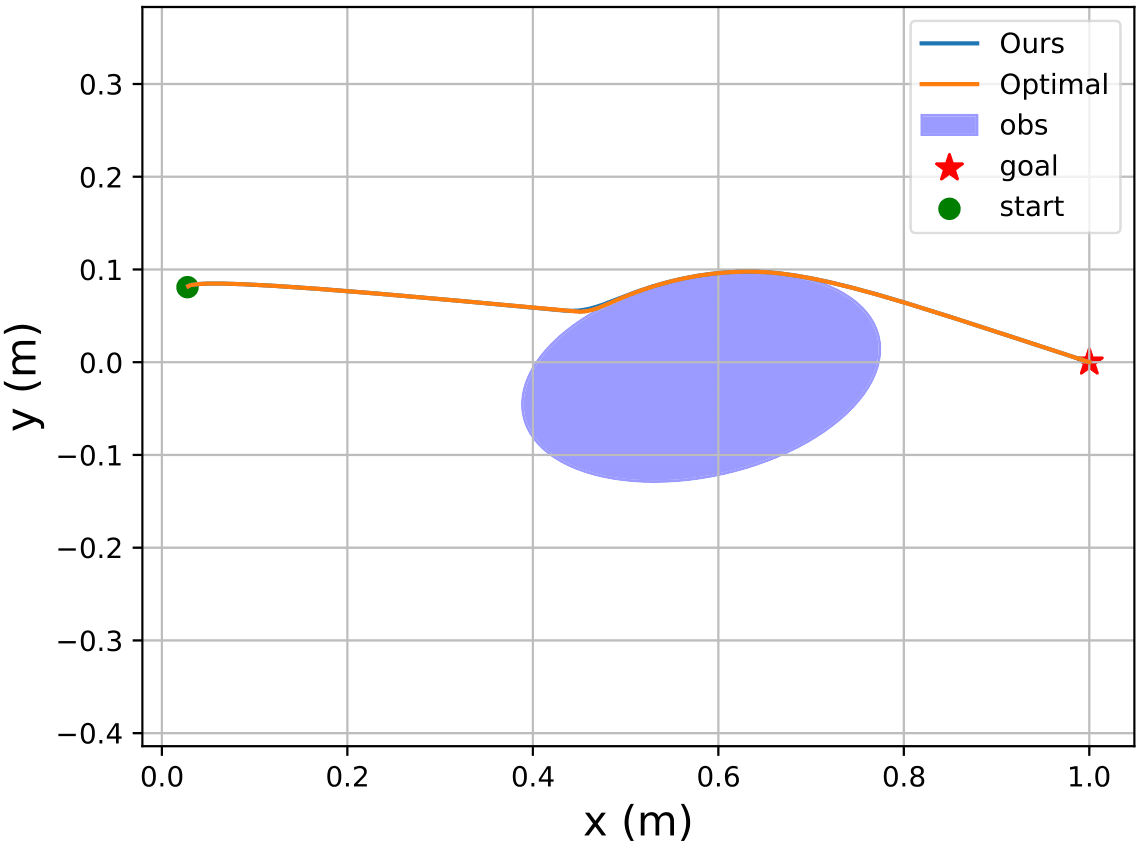}
    	\caption{Environment 2: Trajectory}
    \end{subfigure}
    \begin{subfigure}[r]{0.49\columnwidth}
        \includegraphics[width=\textwidth]{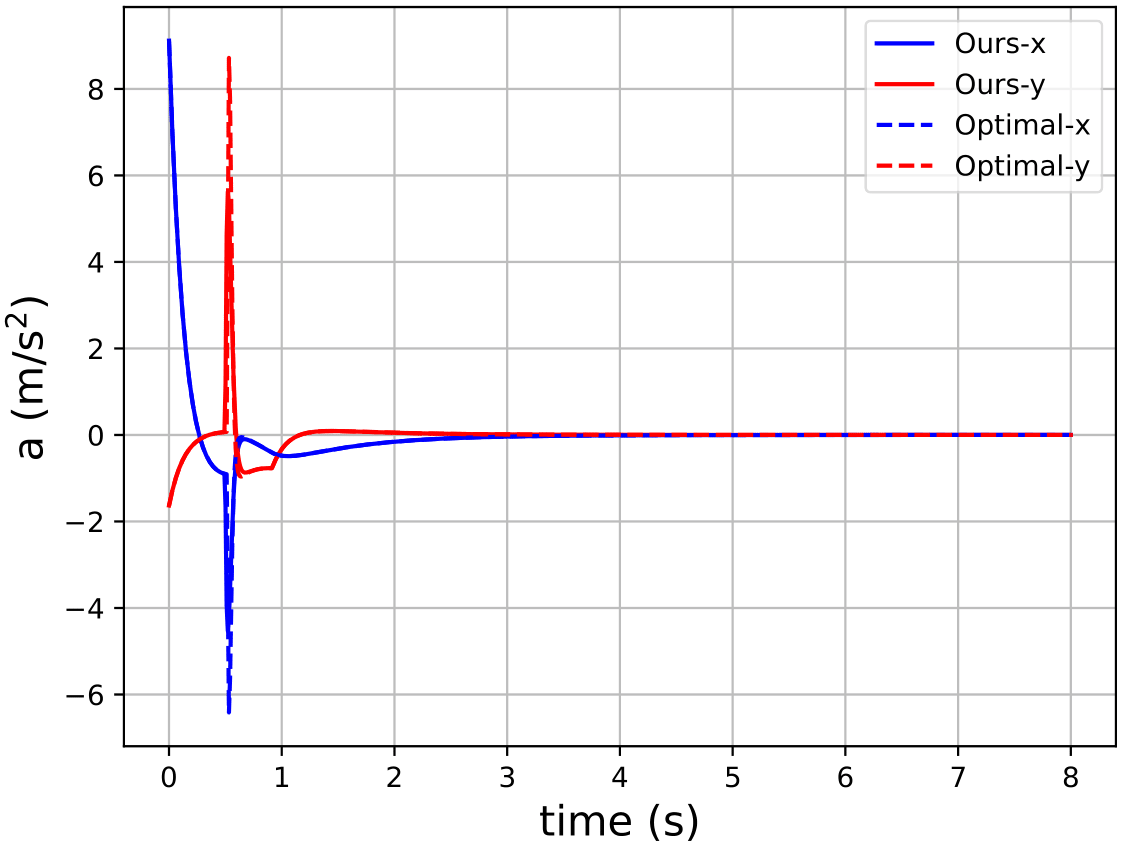}
    \caption{Environment 2: Control input}
    \end{subfigure}
\caption{2D double integrator (n=4, r=2) avoids a randomly generated obstacle in two different environments. The blue trajectory uses the proposed method, and the orange trajectory is the optimal performance reference that was generated by learning specifically for that environment. The corresponding control inputs are shown on the right side.}
\label{fig:2ddi}	
\end{figure}

\subsection{Simulation Setup}
We focus on the collision avoidance problem. 
We set up different environments with different obstacles, which are represented by ellipses:
\begin{equation}
\begin{aligned}
&h(y) = (y-y_c)^{\top}Q(y-y_c) - 1,\\
&y = Cx, Q = R(\theta)^{\top}\Lambda R(\theta),\\
\end{aligned}
\end{equation}
where y is the measurement variable, i.e., the position in Cartesian space. $R(\theta)$ is the rotation matrix defined by the orientation $\theta$ of the ellipse. $y_c$ is the center of the obstacle. $\Lambda$ is a diagonal matrix which represents the size of the obstacle. 
We define the environment information as $e = [y_c, diag(\Lambda)^{\top}, \theta]^{\top} \in \mathbb{R}^{5}$. 
For different environments, we randomly sample $e$ from a Gaussian distribution $\mathcal{P}_{e}$.

We use a linear quadratic regulator as the performance controller and a differentiable ECBF-QP as the safety filter. 
For the $\alpha$-net in the differentiable ECBF-QP, we use a feedforward neural network with two hidden layers. Each hidden layer size is 100. 
Based on Algorithm \ref{algo}, we train each system with $100$ iterations.
In each iteration, we sample $30$ environments, and for each rollout, the simulation time is 8s.
Futhermore, the initial condition of each system is selected randomly within a predefined region.
We use the same loss function for both systems, which is the sum of the distance between each point and the goal location.
\begin{equation}
\begin{aligned}
&\mathcal{L}_{\text{perf}}(\tau, e, x_0) = \sum_{t=0}^{T}\|x_{t} - x_{\text{goal}}\|^2.\\
\end{aligned}
\label{eq:loss_}
\end{equation}
The numerical values of this loss will serve as means to compare performance of different controllers.
Moreover, we train a $\alpha$-net only conditioned on a specific environment to serve as the optimal solution for that specific environment.

\subsection{Double Integrator Experiment}

Two representative validation results for 2D double integrator avoiding an obstacle with the proposed approach are shown in Fig.~\ref{fig:2ddi}, including trajectory and control input.
The start point is chosen randomly, the goal location is at $(1.0 , 0.0)$, and the obstacle is colored as blue.
Moreover, the proposed method is compared with the optimal performance solution, which is obtained by finding the best $K_\alpha$ based on the current environment, i.e., the environment in Fig.~\ref{fig:2ddi}.
In Environment 1, the losses defined in \eqref{eq:loss_} for our method and optimal performance solution are $26.88$ and $26.41$. In Environment 2, the losses are $25.88$ and $25.86$ for ours and optimal performance solution, respectively.
The simulation result shows that the performance of our proposed method is close to the optimal performance solution in terms of the loss function in \eqref{eq:loss_}. Also, our control inputs (solid lines) is similar to the optimal ones (dashed lines) as shown in Fig.~\ref{fig:2ddi} (b) and (d).

\subsection{Quadruple Integrator Experiment}
In Fig.~\ref{fig:2dqi}, we show that our approach can cope with a system with relative-degree four for generalization to novel environments.
\begin{figure}[t]
\centering
    \begin{subfigure}[r]{0.49\columnwidth}
        \includegraphics[width=\textwidth]{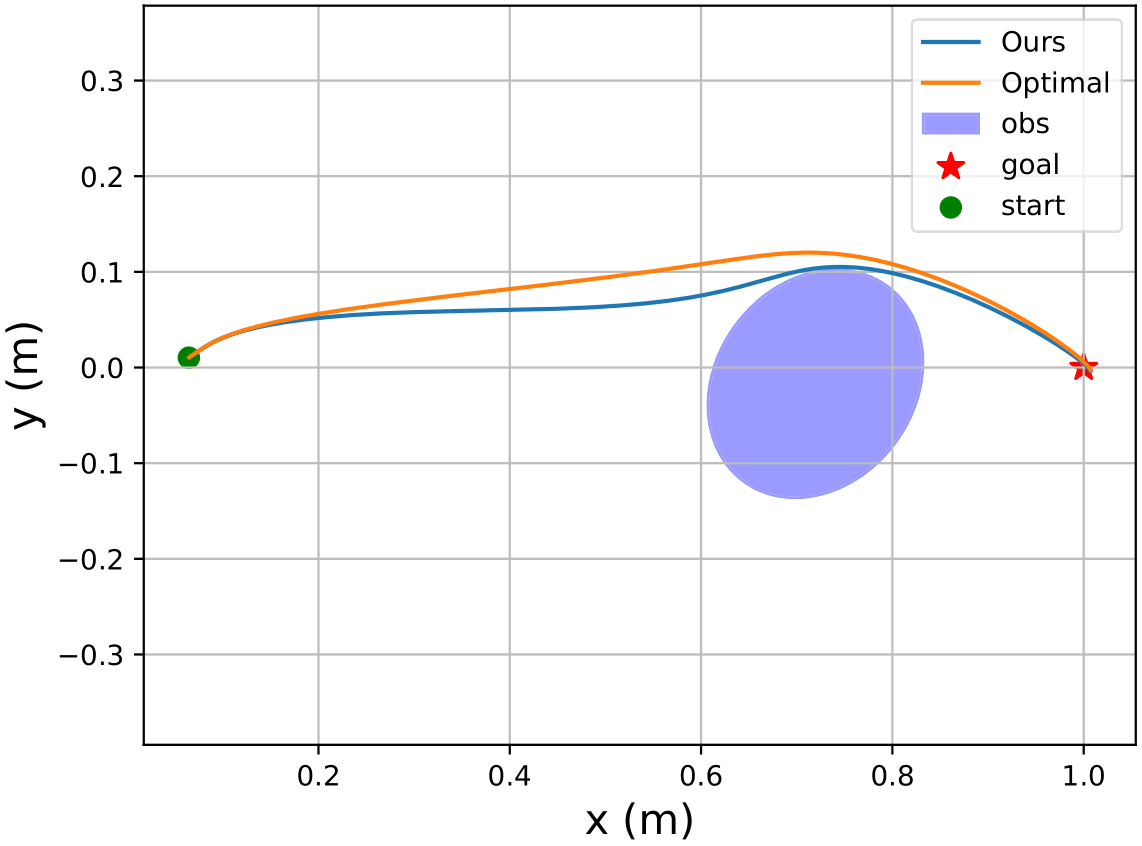}
    	\caption{Experiment 1: Trajectory}
    \end{subfigure}
    \begin{subfigure}[r]{0.49\columnwidth}
        \includegraphics[width=\textwidth]{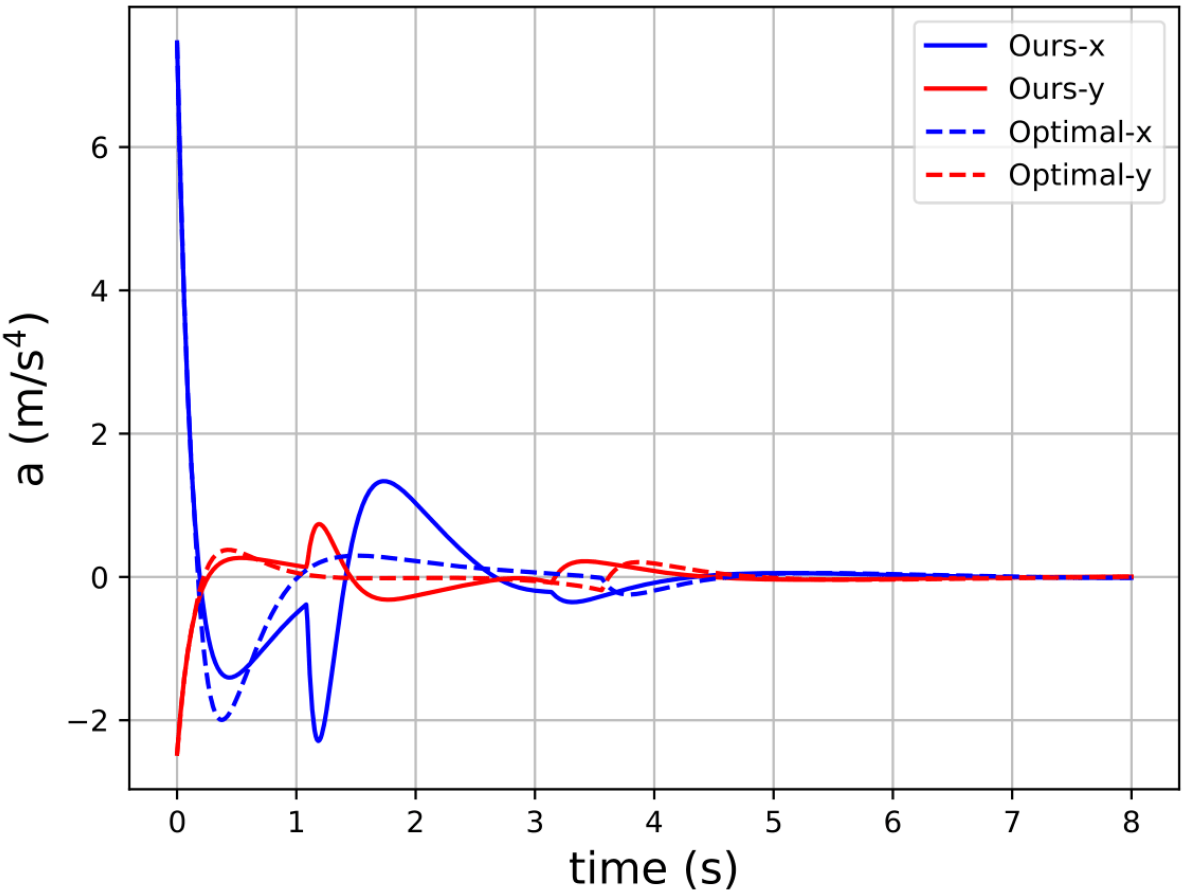}
    \caption{Experiment 1: Control input}
    \end{subfigure}
    \begin{subfigure}[r]{0.49\columnwidth}
        \includegraphics[width=\textwidth]{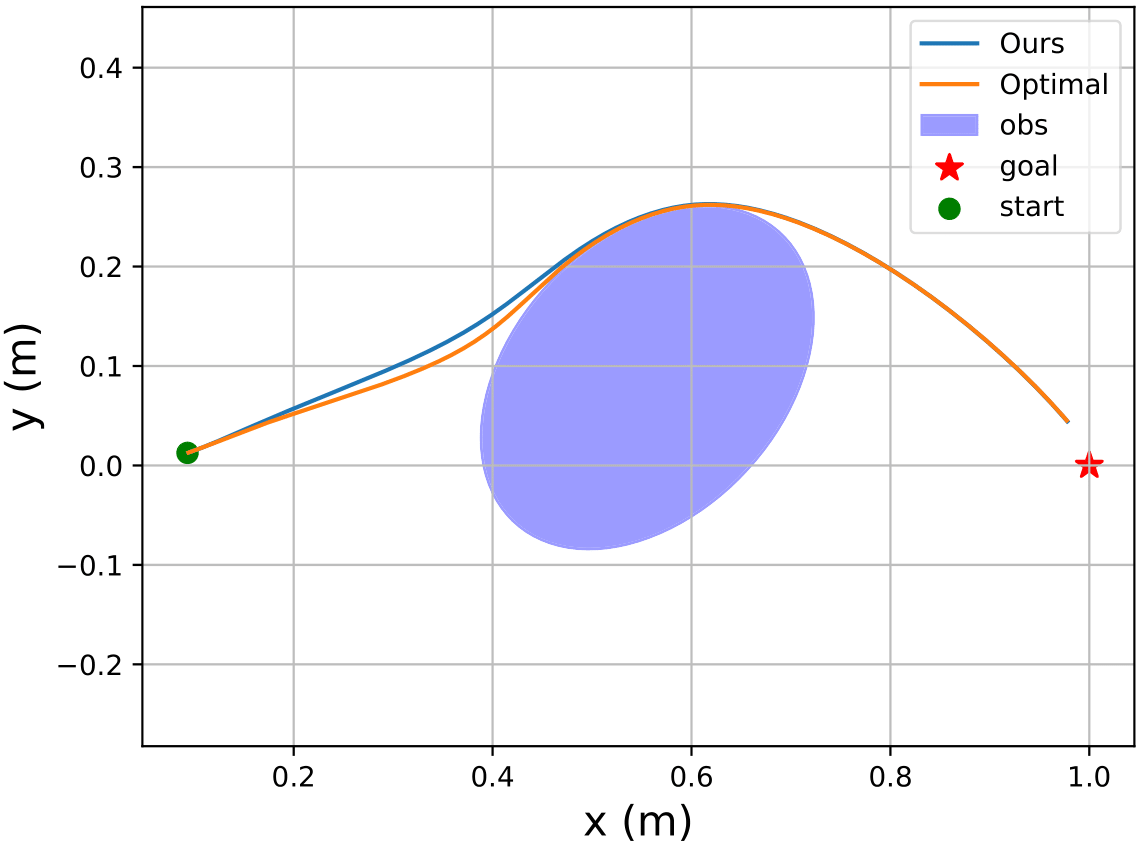}
    	\caption{Experiment 2: Trajectory}
    \end{subfigure}
    \begin{subfigure}[r]{0.49\columnwidth}
        \includegraphics[width=\textwidth]{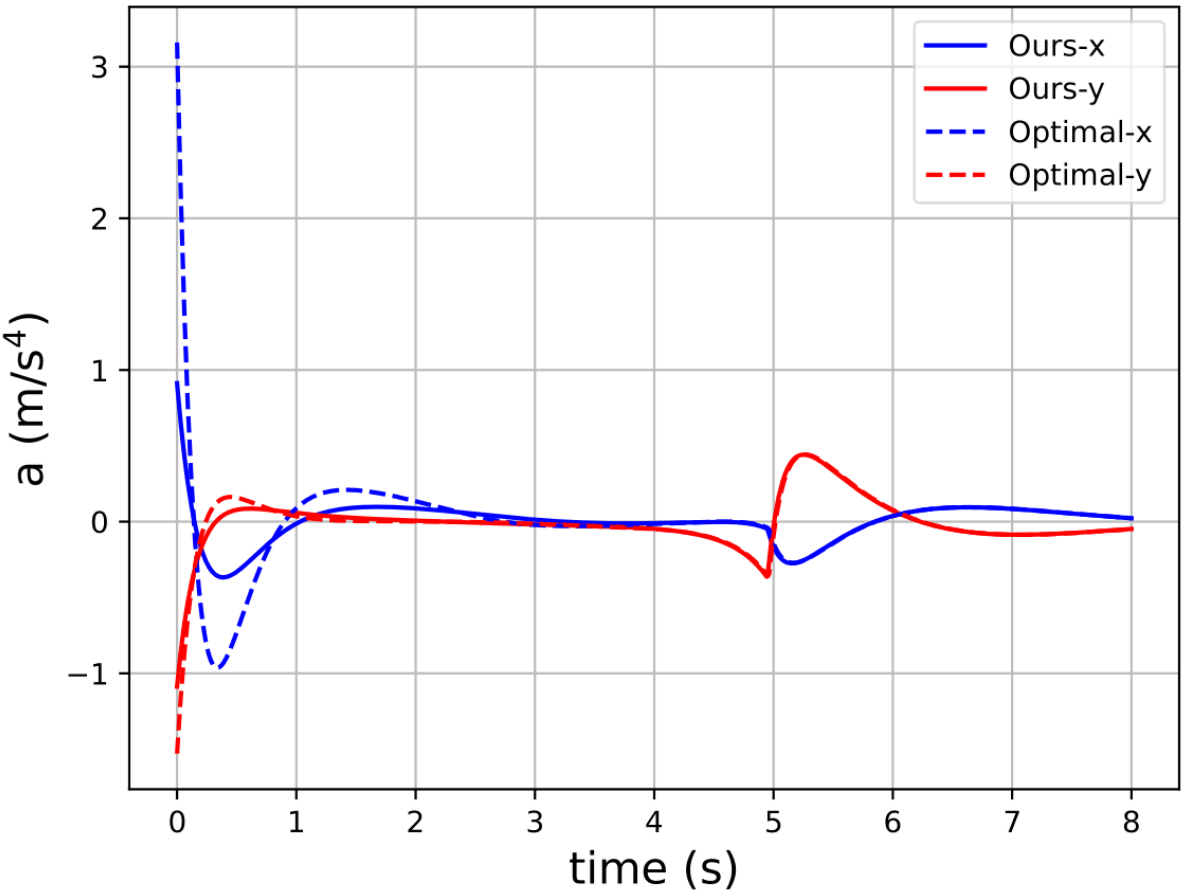}
    \caption{Experiment 2: Control input}
    \end{subfigure}
\caption{2D quadruple integrator (n=8, r=4) avoids a randomly generated obstacle in two different environments. The blue trajectory uses the proposed method, and the orange trajectory is the optimal performance reference that was generated by learning specifically for that environment. The corresponding control inputs are shown on the right side.}
\label{fig:2dqi}	
\vspace{-3mm}
\end{figure}
In both environments, the losses for the optimal performance solution is $26.71$ and $35.03$, whereas the losses for our method is $28.35$ and $35.36$, respectively. We also observe that the control inputs of the proposed method is similar to the optimal performance solution as shown in Fig.~\ref{fig:2dqi} (b) and (d). The results imply that our approach can determine a proper $K_{\alpha}$ given the environment information without any manually tuning process for high relative-degree systems.

\subsection{Multiple Obstacles Experiment}

To further validate the generalization ability, we extend the simulation setup of 2D double integrator from one obstacle to multiple obstacles.
We randomly generate two obstacles and formulate one ECBF constraint for each object accordingly in \eqref{equ:diff-cbf-1}. For each constraint, we use the same learned $\alpha$-net.
In this scenario, the proposed method needs to be able to generalize to more complex environments.
The simulation results of two examples are shown in Fig.~\ref{fig:generalization}.
In Experiment 1, the losses for our method and random valid $K_\alpha$ are $28.11$ and $32.50$, respectively,
and in Experiment 2, the corresponding losses are $23.96$ and $35.51$.
It shows that our approach successfully generalizes to multiple obstacle scenarios and outperforms the baseline by a large margin. 

\begin{figure}
\centering
    \begin{subfigure}[r]{0.49\columnwidth}
        \includegraphics[width=\textwidth]{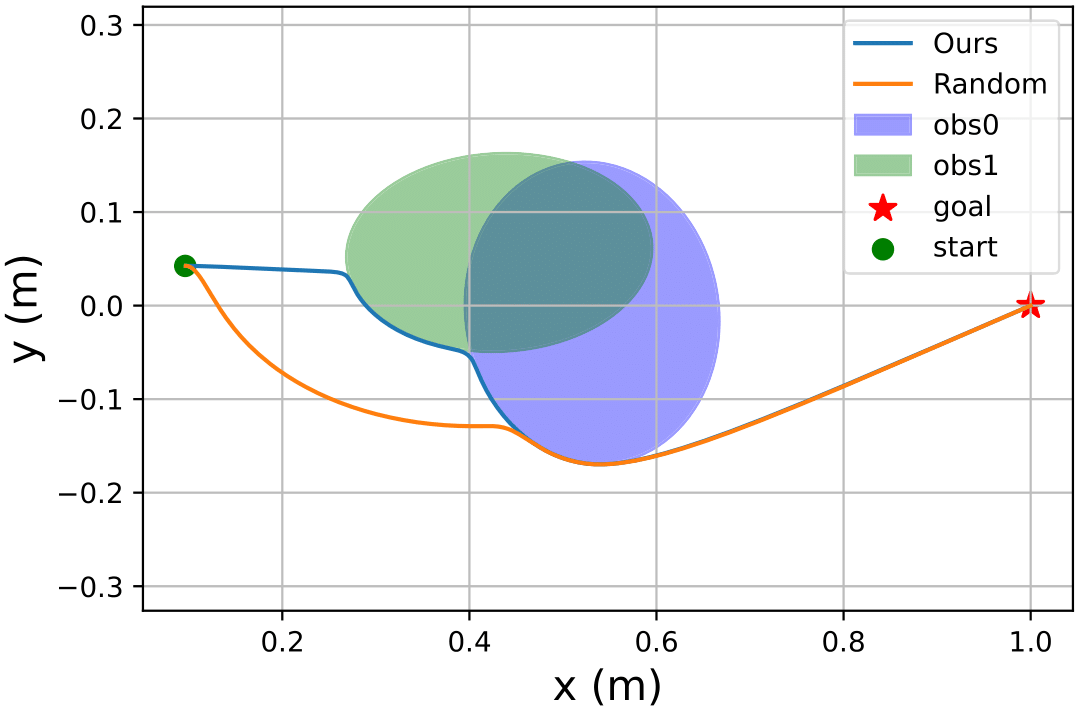}
    	\caption{Experiment 1}
    \end{subfigure}
    \begin{subfigure}[r]{0.49\columnwidth}
        \includegraphics[width=\textwidth]{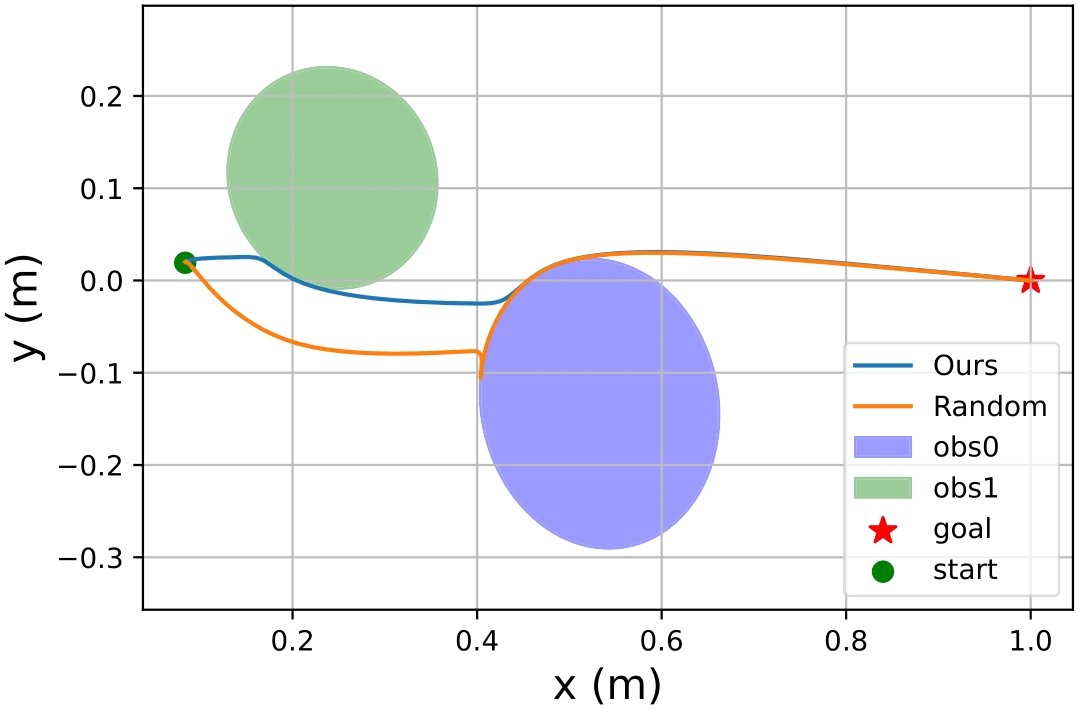}
    \caption{Experiment 2}
    \end{subfigure}
\caption{2D double integrator is able to generalize to novel environments with two randomly generated obstacles, which are not experienced during training. The blue trajectory utilizes the proposed framework, and the orange trajectory uses a random valid $K_{\alpha}$.}
\label{fig:generalization}	
\end{figure}

\subsection{Ablation Study}

We conduct two ablation studies using the 2D quadruple integrator to validate that our proposed design is necessary for generalizing to novel environments.

\subsubsection{\todo{Is the obstacle information indeed useful?}}
The first ablation study is to evaluate whether the obstacle information is necessary.
We train an $\alpha$-net with only one fixed obstacle during training as a baseline and then test it with novel environments.
The resulting trajectories are shown in Fig.~\ref{fig:ablation}(a).
Our proposed method has a loss of $30.26$, while the loss for baseline is $36.79$. This indicates that the obstacle information is necessary as an input to our neural network.

\subsubsection{\todo{Is a larger or smaller valid $K_{\alpha}$ better?}}
In Fig.~\ref{fig:ablation}(b), \todo{we investigate whether a larger or smaller valid $K_{\alpha}$ can achieve a better performance with respect to our loss function.
We scale $K_{\alpha}$ by multiplying the learned eigenvalues $\{p_{i}\}_{i=1, \dots, r}$ with coefficients $3.0$ and $0.5$.}
The losses for our method, $3.0$ scale, and $0.5$ scale are $33.90$, $65.90$, and $37.85$, respectively.
The resulting trajectories demonstrate that the proposed method outperforms the cases with the scales.

\begin{figure}
\centering
    \begin{subfigure}[r]{0.49\columnwidth}
        \includegraphics[width=\textwidth]{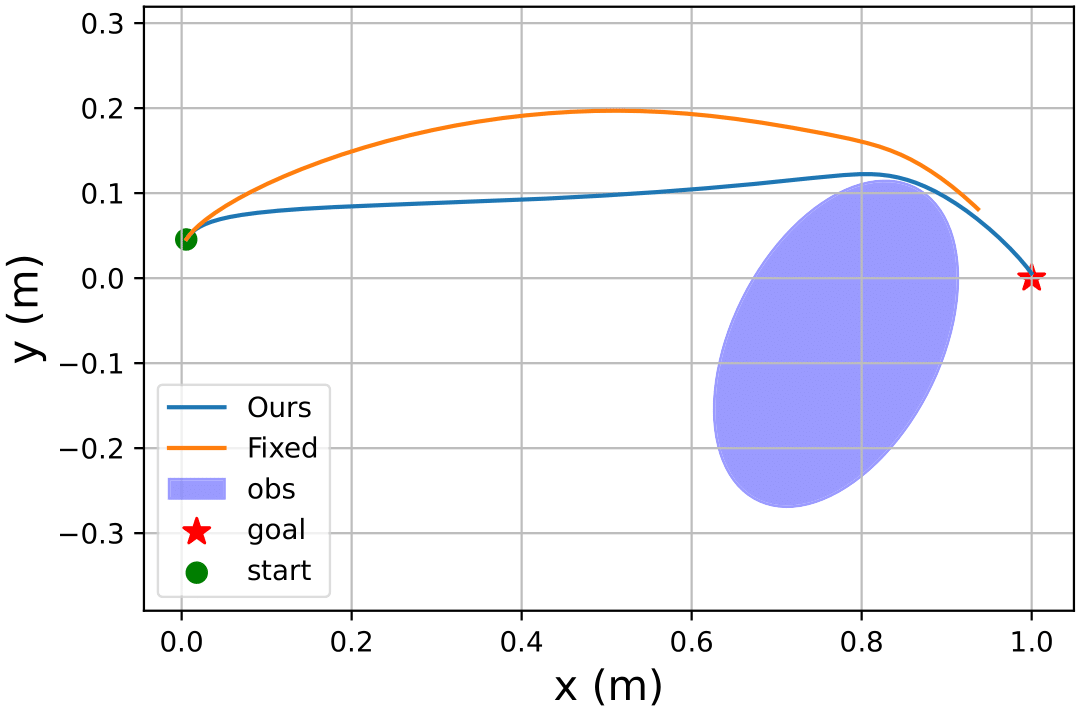}
    	\caption{Ablation study 1}
    \end{subfigure}
    \begin{subfigure}[r]{0.49\columnwidth}
        \includegraphics[width=\textwidth]{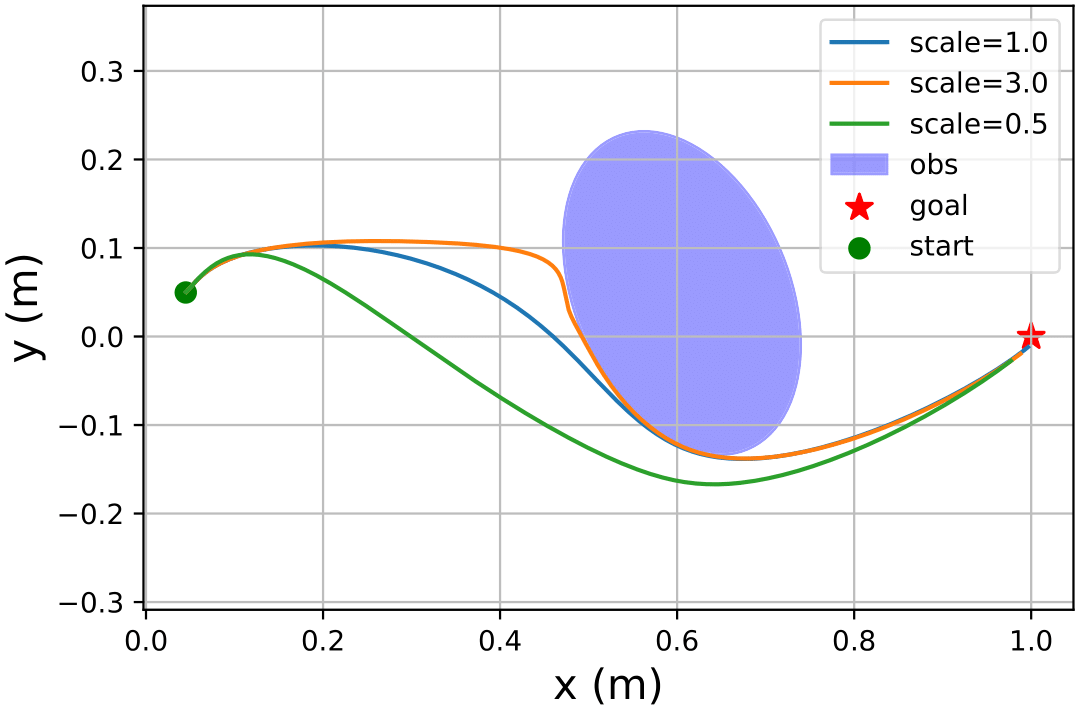}
    \caption{Ablation study 2}
    \end{subfigure}
\caption{Ablation study: (a) obstacle information is necessary for environment generation. The blue curve uses the proposed method, and the orange curve is the baseline with fixed obstacle information; (b) a larger or smaller valid $K_{\alpha}$ does not lead to a better performance. Different colors represent different scales.}
\label{fig:ablation}	
\end{figure}
\section{Discussion}
\label{sec:discussion}

We next provide the analysis of the proposed framework and a discussion on the limits and thoughts on future work.

\begin{table}
    \centering
    \resizebox{\linewidth}{!}{
    \begin{tabular}{|c|c|c|} \hline
       Scenario & Random & Ours \\ \hline
       Double Integrator & 26.8832$\pm$0.7259 & \textbf{26.6774}$\pm$\textbf{0.5998} \\ 
       Quadruple Integrator & 34.2979$\pm$1.1136 & \textbf{32.2187}$\pm$\textbf{0.8840} \\ 
       Multiple Obstacles & 62.2640$\pm$2.2653 & \textbf{40.6258}$\pm$\textbf{3.6791} \\
       \hline
    \end{tabular}
    }
    \caption{Benchmark of our proposed framework in three different scenarios: double integrator, quadruple integrator, and double integrator with two obstacles.}
    \label{tab:benchmark}
    \vspace{-2mm}
\end{table}

We carry out a performance benchmark in three different scenarios: 2D double integrator with one obstacle (Double Integrator), 2D quadruple integrator with one obstacle (Quadruple Integrator), and 2D double integrator with two obstacles (Multiple Obstacles).
We compare our method with random valid $K_\alpha$ using $200$ experiments for each scenario.
The mean and standard deviation of the losses are summarized in Tab.~\ref{tab:benchmark}.
We first conduct 4 subtasks, and each of them consists of 50 experiments. The mean and standard deviation are computed based on these subtasks.
The benchmark for Double Integrator and Quadruple Integrator further validates that the proposed framework is useful for generalization to novel environments.
Moreover, when we extend the learned $\alpha$-net to multiple obstacles, our method shows better results.

In terms of the overall controller design, we assume that a high-level controller is given and fixed in this work, which could be a valid assumption in many applications. Note that the overall performance is determined by both the high-level performance controller and the CBF-QP safety filter.
\section{Conclusion}
\label{sec:conclusion}

In this paper, we presented a learning differentiable safety-critical-control framework using control barrier functions for generalization to novel environments, which uses a learning-based method to choose an environment-dependent $K_\alpha$ in exponential control barrier function.
Moreover, based on the ECBF formulation, the proposed method ensures the forward invariance of the safe set.
We numerically verified the proposed method with 2D double and quadruple integrator systems in novel environments.
Our framework can be easily generalized to different shapes of obstacles and nonlinear dynamics. Also, different representation of environment information such as images and point cloud can be used. There are several interesting future directions. For instance, an integrated end-to-end framework can be designed for training the performance controller and ECBF-QP. Another promising future direction is to test our approach in more general scenarios.

\section*{Acknowledgement}
We would like to thank Ayush Agrawal for his helpful discussions.

\balance







\bibliographystyle{IEEEtranS}
\bibliography{reference.bib}{}
\end{document}